\documentclass[12pt]{article}
\usepackage{amsmath}
\usepackage{graphicx}
\usepackage{enumerate}
\usepackage{natbib}
\usepackage{url} 

\usepackage{amsthm,amsfonts,amssymb,dsfont}
\usepackage{xcolor}
\usepackage[hidelinks]{hyperref}

\newcommand{\blind}{1}

\newtheorem{lem}{Lemma}[section]

\theoremstyle{definition}

\theoremstyle{remark}

\numberwithin{equation}{section}

\begin{document}

\def\spacingset#1{\renewcommand{\baselinestretch}%
{#1}\small\normalsize} \spacingset{1}


\if1\blind
{
	\title{\bf Accurate Estimates of Ultimate 100-Meter Records}
	\author{
		John H.J.\ Einmahl \\
		Department of Econometrics and Operations Research, 
		Tilburg University
		\and
		Yi He\\
		Amsterdam School of Economics, University of Amsterdam}
	\maketitle
} \fi

\if0\blind
{
	\bigskip
	\bigskip
	\bigskip
	\begin{center}
		{\LARGE\bf  Accurate Estimates of Ultimate 100-Meter Records}
	\end{center}
	\medskip
} \fi
\bigskip

\bigskip
\begin{abstract}
We employ the novel theory of heterogeneous extreme value statistics to accurately estimate the ultimate world records for the 100-m running race, for men and for women. For this aim we collected data from  1991 through 2023 from thousands of top athletes, using multiple fast times per athlete. We consider the left endpoint of the probability distribution of  the running times of a top athlete and define the ultimate world record as the minimum, over all top athletes, of all these endpoints.
For men we estimate the ultimate world record to be 9.56 seconds. More prudently, employing this heterogeneous extreme value theory we construct an accurate asymptotic 95\% lower confidence bound on the ultimate world record of 9.49 seconds, still quite close to the present world record of 9.58.
For the women's 100-meter dash our point estimate of the ultimate world record is 10.34 seconds, somewhat lower than the world record of 10.49. The  more prudent 95\% lower confidence bound on the women's ultimate world record is 10.20. 
\end{abstract}

\noindent%
{\it Keywords:}  Endpoint estimation; extreme value statistics; heterogeneous data;  100-m running.
\vfill

\newpage
\spacingset{1.2} 

\section{Introduction}

Athletics, often referred to as the ``mother of all sports'', consists of many events involving running, jumping, and throwing. In particular,  the 100-meter dash at the Olympic Games or World Championships gets a lot of attention all over the world. 
In this paper we therefore focus on the 100 meters for both
men and women. We would like to answer the question how fast the 100m can be run, that is, we are interested in the ultimate world record under the current conditions. 

A straightforward approach to answer this question consists 
of plotting the historical progression of world records and extrapolating into the future.
However, such a method relies on few data and therefore yields inaccurate estimates. Moreover, it fails to address the  question of the fastest possible time achievable “tomorrow” instead of  that in the distant future.

To achieve a  much more precise answer our approach uses extreme value statistics and is based on thousands of data. That approach is also used in, e.g., Einmahl and Magnus
(2008) to estimate ultimate world records, but here we employ the recent, more
refined methods based on heterogeneous extreme value statistics introduced  in He and
Einmahl (2024). The heterogeneous data approach  leads to lower variances
and hence more accurate estimates of ultimate records. This improvement
is substantially enhanced by the fact that we enlarge the sample size very much by using
multiple running times of one athlete instead of only the personal best time.

We model these multiple running times for the \(\ell\)-th athlete as independent and identically distributed samples from some distribution function with a positive left endpoint \(\alpha_\ell\); times from different athletes are assumed to be independent. Our objective is to accurately estimate the ultimate world record, defined as \(\alpha = \min_\ell \alpha_\ell\), and to provide sharp 95\% lower confidence bounds for \(\alpha\).
Our findings reveal that the current world records, held by Usain Bolt and Florence Griffith Joyner, are rather close to the estimated ultimate world records, highlighting their extraordinary achievements.

This paper is organized as follows. In Section 2 the data that we use in the analysis are described. In Section 3 we present the extreme value theory for heterogeneous data. The results for the 100-m running are presented in Section 4. Section 5 provides a summary and conclusion and finally in the Appendix a novel   lemma  to quantify  tail heterogeneity is stated and   proved. 

\section{The 100-meter Sprint Race Data}\label{sec:data}
We obtained the annual best performances of 100-meter athletes from the ``All-Time Top Lists'' on the World Athletics website:  
\begin{center}  
	\url{https://worldathletics.org.}  
\end{center}  
For each athlete, we used multiple best annual records, with up to five observations per athlete, from the  years 1991 through 2023. The analysis was conducted separately for men and women, yielding a total of 5618 male athletes and 2528 female athletes. Table \ref{tab:summary} shows some summary statistics of our dataset.
\begin{table}[!h]
	\centering
	\caption{Comparison of our (EH) data with those in \cite{EM2008}}
	\label{tab:summary}
\begin{tabular}{cccccccccc} \hline 
	~&\multicolumn{4}{c}{Male} &~ &  \multicolumn{4}{c}{Female}  \\\cline{2-5}\cline{7-10} 
	&	Athletes & Records  &    Best& 970th &~& Athletes & Records  &  Best & 578th  \\ 
	EH&5618 & 25244 & 9.58 & 10.09 & ~ & 2528 & 11654 & 10.54 & 11.09 \\ 
	EM2008&970 & 970 & 9.78 & 10.30 & ~ & 578 & 578 & 10.49 & 11.38 \\ 
	\hline 
\end{tabular}
\end{table}

This dataset includes approximately five times as many athletes as \cite{EM2008},  abbreviated as EM2008, in which only the personal best time of each athlete was considered. Overall, our dataset consists of 25244 observations for male athletes and 11654 observations for female athletes, making it more than  20 times as large as that   of EM2008.

To mitigate rounding errors, we followed  EM2008 and smoothed  equal times into  the corresponding rounding interval. 
For example, if there are \( m \) observations of 11.05 seconds across the female athletes, these \( m \) times are smoothed over the interval \( (11.045, 11.055) \) using the formula:  
\begin{equation*}
	 11.045 + 0.01 \frac{2j - 1}{2m}, \quad j = 1, \ldots m,	
\end{equation*}
(where the time with $j=1$  corresponds to the lowest measured wind speed, etc.). Also, for the analysis we will convert time measurements (in seconds) into  corresponding average speed values (in kilometers per hour). However, the final results  will be converted back to time measurements.

\section{Methodology in a General Framework}
\subsection{Detecting Heterogeneity using Multiple Records}\label{het}
Our statistical methodology is based on extreme value theory. In the next subsection we will present that general theory for heterogeneous data,  but first we describe specifically how to quantify the  tail heterogeneity when using multiple data for each athlete. 
Consider possibly heterogeneous, independent speed data $X_1^{(n)}, \ldots, X_n^{(n)}$, where $n$ is the total sample size. These data belong to $p$ athletes. The $\ell$-th athlete ($\ell =1, \ldots, p$) has $m_{\ell}\geq 1$ speed data, hence  $n = \sum_{\ell=1}^{p} m_{\ell} \geq p$. We assume $\max_{\ell=1, \dots, p}m_\ell$ stays bounded in the asymptotic theory.
The indices are such that the records for each athlete are grouped consecutively: for the $\ell$-th athlete the data are $X_{r_{\ell}+1}, \dots, X_{r_{\ell}+ m_{\ell}}$, for some index  $r_{\ell}$. 
Denote the distribution function of $X_i^{(n)}$ by $F_{ni}, i=1, \dots,  n$, and  note that there are at most $p$ different distribution functions. The average distribution function is  given by
$F_n= \frac{1}{n}\sum_{i=1}^{n} F_{ni}$
and is assumed to be continuous.

According to \cite{HE2024}, the  heterogeneity in the right tail is   characterized by the  function:
\begin{equation}\label{lambda-func:general}
	\lambda(u)= \lim\limits_{n \rightarrow \infty} \frac{1}{k} \sum_{i=1}^{n} P\left(U_i^{(n)} < \frac{k}{n}\right) P\left( U_{i}^{(n)} < \frac{k}{n} \frac{1}{u} \right),~U_i^{(n)} := 1 - F_n(X_i^{(n)}), \quad u>0, 
\end{equation}
with $k = k(n)$  any intermediate sequence, that is:
\begin{equation}\label{k:p}
	k \to \infty, \quad k/n \to 0, \quad \text{as } n \to \infty.
\end{equation}
It is worth noting that $\lambda$ is either positive for all $u > 0$ or zero  for all $u > 0$. For identically distributed (homogeneous) data, $\lambda(u)= 0$ for all $u > 0$. In general, it is also possible for heterogeneous data to have $\lambda \equiv 0$. 
In such cases, the basic asymptotic theory  for homogeneous data remains valid; see, e.g., \cite{EHZ2016}. However, when $\lambda(u) > 0$, the asymptotic variance of the  endpoint estimator changes, making novel methods necessary.

Suppose $m_{\ell} \geq 2$ for each $\ell $.  (If this is not the case for relatively  few values of $\ell$,  drop these records or  duplicate them to ensure that $m_{\ell} \geq 2$ for all athletes after this adjustment.) 
We then obtain an alternative expression of \eqref{lambda-func:general}:
\begin{gather}	\label{lambda-func-defn:estimable}
	\lambda(u) 
	=\lim\limits_{n \rightarrow \infty} \frac{1}{k} \sum_{\ell=1}^{p} \Lambda_{\ell}(u),~
	\Lambda_{\ell}(u)= 	\frac{1}{m_{\ell} - 1} \sum_{1 \leq j_1 \neq j_2 \leq m_{\ell}} P\left( U_{r_{\ell}+j_1}^{(n)} < \frac{k}{n}, U_{r_{\ell}+j_2}^{(n)} < \frac{k}{n} \frac{1}{u} \right).
\end{gather}
Now, taking the empirical analogue of \eqref{lambda-func-defn:estimable} gives us a novel, nonparametric estimator of $\lambda(u)$. Denote the (ascending) ranks of the speeds $X_{i}^{(n)}$ as $R_1, \ldots, R_n$. 
For the $\ell$-th athlete's i.i.d.\ data of size $m_{\ell} $, $X_{r_\ell+1}, \ldots, X_{r_\ell+m_{\ell}}$, we estimate $\Lambda_{\ell}(u)$ by
\begin{equation}\label{eqn:estimator-lambda-individual}
	\widehat{\Lambda}_{\ell}(u) = \frac{1}{m_{\ell}-1} \sum_{1 \leq j_1 \neq j_2 \leq m_{\ell}} \mathds{1} \left[R_{r_\ell+j_1} > n-k, \, R_{r_\ell+j_2} > n-\frac{k}{u} \right].
\end{equation}
From this, $\lambda(u)$  is estimated by:
\begin{equation}\label{eqn:estimator-lambda}
	\widehat{\lambda}(u) = \frac{1}{k} \sum_{\ell=1}^{p} \widehat{\Lambda}_\ell(u),\quad u>0.
\end{equation}
We need and will show the uniform consistency of $\widehat{\lambda}$, that is, for a sequence $k = k(n)$ satisfing (\ref{k:p}) and  for each $\delta>0$, as $n \to \infty$,
	\begin{equation}\label{uclh}
		\sup_{u \geq \delta}\, \left| \,\widehat{\lambda}(u) - \lambda(u) \, \right| \xrightarrow{P} 0,
	\end{equation}
where $\xrightarrow{P}$ denotes convergence in probability.
    A generalization of this novel result will be formulated precisely and proved in the Appendix.

\subsection{Accurate Confidence Bounds on Ultimate Records}\label{sec:evt-estimator}
In this section, we describe how to improve statistical inference by incorporating heterogeneity. We employ   the extreme value theory for heterogeneous data in \cite{HE2024}. There the first-order condition in extreme value theory, which is classically applied to homogeneous data,  accounts for heterogeneous data: the average survival function $T_n = 1 - F_n$ of the speed data converges to a limiting survival function $T = 1 - F$ in the tail, where $F$ has an extreme value index $\gamma < 0$. More precisely:
\begin{enumerate}[(a)]
	\item There exist sequences $a_t$ and $b_t$ such that:
	\begin{equation}\label{eqn:DoA-limit-survival}
		\lim\limits_{t \to \infty} t T(a_t x + b_t) = (1 + \gamma x)^{-1/\gamma}, \quad x < -1/\gamma.
	\end{equation}
	\item For all sequences $t = t(n) \uparrow x^*:=\sup\{x : F(x) <1\} < \infty$, with $n T(t) \to \infty$,
	\begin{equation*}
		\frac{T_n(t)}{T(t)} \to 1.
	\end{equation*}
	\item There exists a positive constant $M > 0$ such that for all sufficiently large $t$ and $n$, $T_n(t) \leq M T(t)$.
\end{enumerate}

Define the right endpoint of the speed distribution $\widetilde{F}_{\ell}$ for the $\ell$-th athlete by
\begin{equation*}
x_{p,\ell}^*:=\sup\{x:\widetilde{F}_{\ell}(x)<1\}.
\end{equation*}
Then, the right endpoint of the average distribution  $F_n$ is equal  to their maximum, namely
\begin{equation*}
	x_n^*:=\sup\{x:F_n(x)<1\}=\max_{ \ell=1, \ldots,  p}x_{p,\ell}^*.
\end{equation*}
Observe that $x_n^*=x^*$ for homogeneous data  $T_n\equiv T$. In general, for heterogeneous data, it follows from the assumptions that:
\begin{equation*}
	x_n^*\leq x^*,\quad\text{and}\quad x_n^*\to x^* \quad \text{as}~n\rightarrow\infty.
\end{equation*}
Hence although it is allowed that $x_n^*<x^*$, a further assumption in \cite{HE2024} entails that  $x_n^*\to x^*$ fast enough to ensure that they are asymptotically indistinguishable. We define the ultimate (speed) world record as $x_n^*$, but the estimates and confidence bounds will be the same as if we would define it as $x^*$.

Now, take an intermediate sequence $k=k(n)$ as in \eqref{k:p}, with $k\in\{1, \ldots, n-1\}$. We use the classical estimator of the right endpoint in the homogeneous case, based on the moment estimator $\widehat{\gamma} = \widehat{\gamma}_{M}$ in  \cite{DEdH1989}. The moment estimator is defined as
$$\widehat{\gamma}_{\operatorname{M}} = 
	 M_n^{(1)} +  1 - \frac{1}{2} \left(1 - \frac{(M_n^{(1)})^2}{M_n^{(2)}} \right)^{-1}=: M_n^{(1)} +  1 -V_n  $$
    where
\begin{equation*}
	M_n^{(r)} = \frac{1}{k} \sum_{i=0}^{k-1} \left(\log X_{n-i,n} - \log X_{n-k,n} \right)^{r}, \quad r = 1,2,
\end{equation*}
and $X_{1,n} \leq \ldots \leq X_{n,n}$ are the order statistics of the speed data $\{X_i^{(n)}\}$,
    and the endpoint estimator is defined as 
    \begin{equation}\label{eqn:endpoint-estimator}
    	\widehat{x}^*=X_{n-k,n}\left(1-\frac{M_n^{(1)}
    V_n}{\widehat{\gamma}_{\operatorname{M}}}\right). \end{equation}

Under appropriate regularity conditions, in  \cite{HE2024}  
it is shown that 
$$\left(\frac{\widehat{\gamma}_{\operatorname{M}}^2}{M_n^{(1)}
V_n}-\widehat{\gamma}_{\operatorname{M}}\right)\sqrt{k}(\log\widehat{x}^* - \log x^*_n)\stackrel{d}{\to} N\left(0,\sigma^2_{\text{iid}}(\gamma)(1 - \Delta)\right), \mbox{ as } n\to \infty,$$
where $\xrightarrow{d}$ denotes convergence in distribution. 
Here $$\sigma^2_{\text{iid}}(\gamma)=\frac{\left(1-\gamma \right)^2(1-3\gamma+4\gamma^2)}{(1-2\gamma)(1-3\gamma)(1-4\gamma)}$$ is the variance for i.i.d.\ data (used in EM2008), $\Delta$ quantifies the relative variance reduction due to heterogeneity  given by
\begin{equation}\label{eqn:var-loss-endpoint}
	\Delta = w_0(\gamma)\lambda(1) + w_1(\gamma)m_{\lambda}(-\gamma) + w_2(\gamma)m_{\lambda}(-2\gamma),
\end{equation}
where
\begin{equation*}
	\begin{bmatrix}
		w_0(\gamma) \\
		w_1(\gamma) \\
		w_2(\gamma)
	\end{bmatrix}
	=
	\frac{1}{1 - 3\gamma + 4\gamma^2}
	\begin{bmatrix}
		(1 - 2\gamma)(1 - 3\gamma)(1 - 4\gamma) \\
		-2\gamma(1 - \gamma)(1 - 4\gamma) \\
		8\gamma(1 - 2\gamma)^2
	\end{bmatrix}, \quad
	w_0(\gamma) + w_1(\gamma) + w_2(\gamma) = 1,
\end{equation*}
and, with $\lambda$ as in \eqref{lambda-func:general}, 
\begin{equation*}
	m_{\lambda}(x) = (1 + x) \int_{0}^{1} u^x \lambda(u) \, du.
\end{equation*}
 The variance reduction can be consistently estimated with
\begin{equation}\label{Delta}
	\widehat{\Delta} = w_0(\widehat{\gamma}_{\operatorname{M}})\widehat{\lambda}(1) + w_1(\widehat{\gamma}_{\operatorname{M}})m_{\widehat{\lambda}}(-\widehat{\gamma}_{\operatorname{M}}) + w_2(\widehat{\gamma}_{\operatorname{M}})m_{\widehat{\lambda}}(-2\widehat{\gamma}_{\operatorname{M}}),
\end{equation}
where  $m_{\widehat{\lambda}}$ is obtained by estimating $\lambda$ with $\widehat \lambda$ in (\ref{eqn:var-loss-endpoint}).  This allows for the construction of  asymptotically correct confidence intervals, which are typically  substantially narrower than those based on  i.i.d.\ data.

\section{Ultimate World Records}
We now apply the estimators from the previous section to the 100-meter sprint race data discussed in Section \ref{sec:data} in order to address the following questions for males and females:
\begin{itemize}
	\item[$\bullet$] How tail heterogeneous are the 100-meter athletes?\end{itemize}
    and  foremost 
    \begin {itemize}
	\item[$\bullet$] How fast could  the  100-meter potentially be run now?
\end{itemize}

\subsection{Top Athletes are Heterogeneous}
\begin{figure}[!h]
	\centering
	\includegraphics[width=1\linewidth]{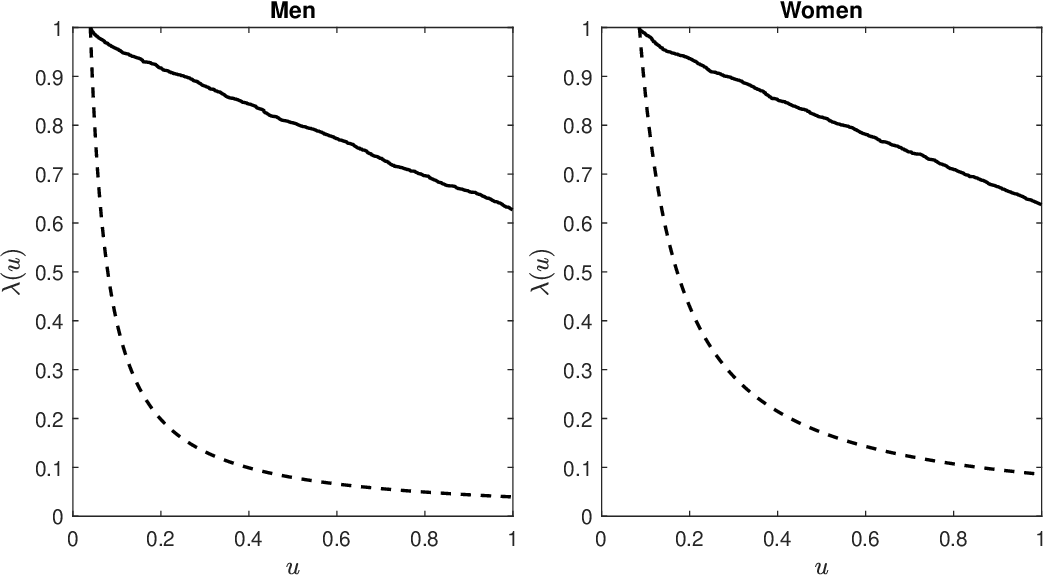}
	\caption{(Solid) Estimates of $\lambda(u)$ on $(0,1]$; (Dashed) Expected values of the estimators if the data were identically distributed.}
	\label{fig:lambdafuncplot}
\end{figure}

Figure \ref{fig:lambdafuncplot} displays the estimated function $\widehat{\lambda}(u)$, $0<u\leq 1$, represented by solid lines, based on multiple records from all athletes, for $k=1000$. (For this purpose we removed athletes with single records: 71 men (0.28\% of the data) and 20 women (0.17\% of the data).)   If the observations were homogeneous and hence exchangeable, one would expect the estimates to align closely with their expected values:  
\begin{equation*}
	\min\left( \frac{\lceil k/u\rceil -1}{n-1}, 1\right),\quad 0<u\leq 1.
\end{equation*}
Clearly this is not the case.
In contrast, our estimates decay almost linearly, showing similar patterns for men and women. The estimates of  $\lambda(1)$ yield the large values 0.63 for males and  and 0.64  for females, highlighting significant heterogeneity across athletes.

\subsection{Current World Records are Close to the Ultimates}
Figure \ref{fig:linearplotmaleandfemale} demonstrates how extreme value theory can be applied to extrapolate beyond the current 100-meter sprint records. The $x$-axis represents a power transformation of the rank of the time record (where the fastest time is ranked 1, and so on), raised to the exponent $-\gamma$. Here, $\gamma$ denotes the extreme value indices, which are estimated to be approximately $-0.20$ for men and $-0.17$ for women, based on the top $k/n \approx 5\%$ of observations. The $y$-axis represents the smoothed speed in km/h.

\begin{figure}[!h]
	\centering
	\includegraphics[width=1\linewidth]{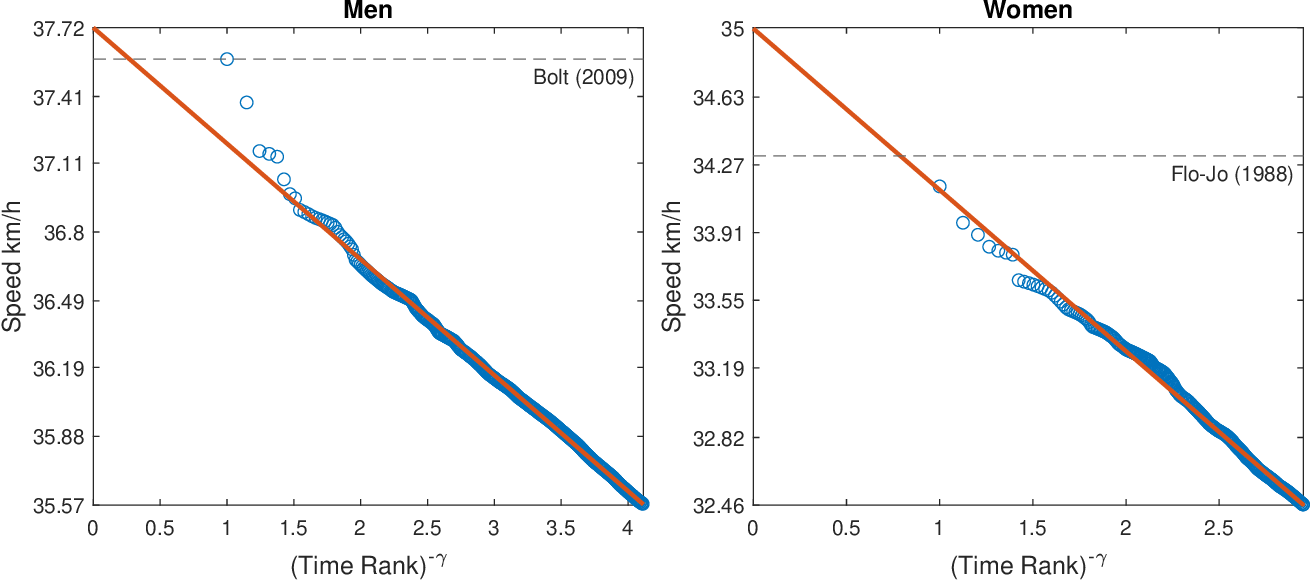}
	\caption{Linear extrapolation using extreme value theory. The \(x\)-axis represents the record rank raised to the power $-\widehat{\gamma}$, with solid lines showing fitted models and intercepts estimating the speed limit. Dashed lines mark current world records, noting that  Florence Griffith-Joyner's 1988 record predates our sample period.}
	\label{fig:linearplotmaleandfemale}
\end{figure}

To understand why the scatter plots exhibit a linear pattern, consider a large threshold $t = t_n=n/k$, which leads to the lower limit of the $y$-axis. The conditions 
 from Section \ref{sec:evt-estimator} provide the following approximations:  
\begin{equation}\label{eqn:DoA-limit-survival-2}
	(t T_n(a_t x + b_t))^{-\gamma}\approx(t T(a_t x + b_t))^{-\gamma} \approx 1 + \gamma x, \quad x > 0.
\end{equation}  
By substituting a large speed, say $z = a_t x + b_t$, we obtain:  
\begin{equation*}
	(t T_n(z))^{-\gamma} \approx 1 + \gamma \frac{z - b_t}{a_t},
\end{equation*}  
or equivalently,  
\begin{equation*}
	z \approx \frac{a_t}{\gamma} 
    k^\gamma (nT_n(z))^{-\gamma} + \left( b_t - \frac{a_t}{\gamma} \right).
\end{equation*}  
Approximating $nT_n(z)$   with $\sum_{i=1}^n \mathds{1}[X_i \geq z]$ as the rank of the time record (instead of the speed), and substituting the parameters $\gamma$, $a_{n/k}$, and $b_{n/k}$ with their estimators provided in \cite{HE2024}, the above expression simplifies to:  
\begin{equation*}
	\text{Speed} \approx \frac{\widehat a_{n/k}}{\widehat\gamma} k^{\widehat \gamma} \cdot (\text{Time Rank})^{-\widehat{\gamma}} + \widehat{x}^*, \quad 
    \widehat{x}^* = \widehat{b}_{n/k} - \frac{\widehat{a}_{n/k} }{\widehat \gamma},
\end{equation*}  
where $\widehat{x}^*$ is the endpoint estimator defined in \eqref{eqn:endpoint-estimator}. This explains the approximate straight line (linear pattern) discussed above.  Observe that setting the Time Rank equal to 0, yields the endpoint estimator, the upper limit of the $y$-axis. 

For both men and women, the scatter plots in Figure \ref{fig:linearplotmaleandfemale} demonstrate a clear linear pattern, aligning well with these extreme value approximations. The solid lines represent the fitted relationships obtained using our estimators. Importantly, the intercepts of these lines yield estimates for the speed limits: 37.72 km/h for men and 35.00 km/h for women, which correspond to time limits of 9.54 seconds and 10.29 seconds for the 100-meter sprint, respectively.

To assess the sensitivity of our analysis to the choice of the number of tail observations  used in the estimation, Figures \ref{fig:fitk} present our estimates of the extreme value index (left) and the ultimate record in seconds (right) as functions of $k$, ranging from 3\% to 7\% of the total number of observations, for male and female athletes, respectively. The dotted lines represent the estimated values. Overall, the results show that the estimates remain relatively stable across the range of $k$, with some instability for small $k$  for the women's results; however, the estimates stabilize quickly as $k$ increases.

\begin{figure}[!h]
	\centering
	\includegraphics[width=1\linewidth]{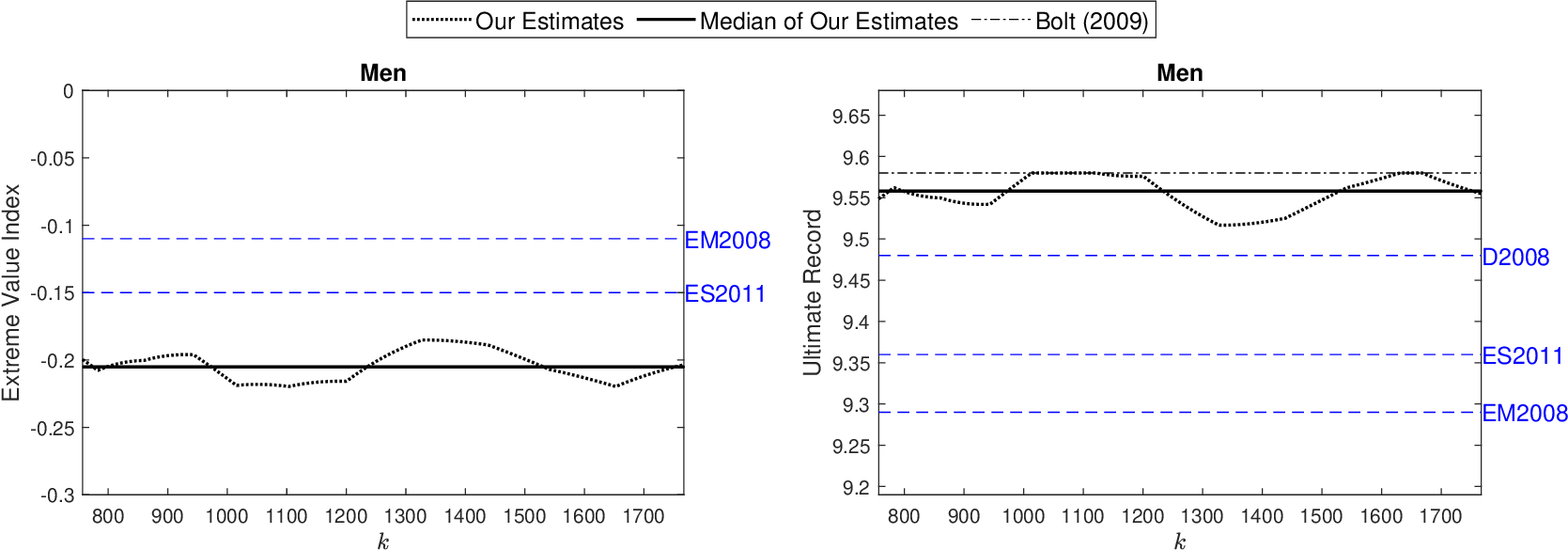}
	\includegraphics[width=1\linewidth]{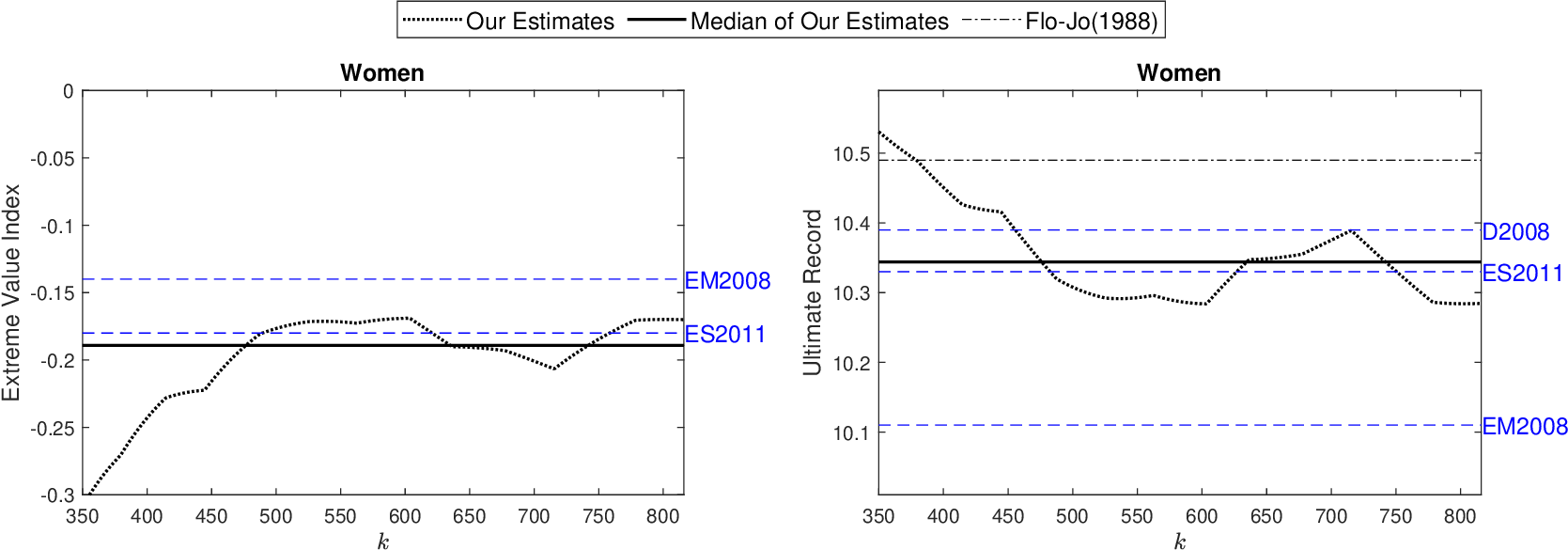}
	\caption{Estimates of the extreme value index (left) and ultimate record (right) are shown as dotted lines varying with \(k\); the solid line is their median. Dashed lines represent earlier estimates, while the dash-dotted line marks current world records, with Florence Griffith-Joyner's 1988 record predating our sample period.}
	\label{fig:fitk}
\end{figure}

To summarize, the median (solid line) of our estimates for the ultimate records are \textbf{9.56} seconds for men and \textbf{10.34} seconds for women. Unlike the previous estimates in EM2008, \cite{ES2011}, and \cite{D2008}, as shown in Figure \ref{fig:fitk}, our findings indicate that the current world record for men is very close to the human limit. Specifically, Usain Bolt's record is nearly ultimate, with potential improvements limited to approximately 0.02 seconds. In contrast, there is a larger margin for improvement for female athletes, with potential gains of up to 0.15 seconds. Notably, we find that Florence Griffith-Joyner's longstanding record, though yet to be broken, remains within the human limit.

However, it is essential to consider the statistical uncertainties associated with these estimates. Figure \ref{fig:varlosskep} illustrates the percentage of variance reduction achieved by accounting for heterogeneity,  estimated with $\widehat\Delta$  in (\ref{Delta})  in the previous section, across different choices of $k$. For both men and women, the results show a consistent reduction of approximately 35\% in the asymptotic variance, compared to the (incorrect) variance used in i.i.d.\ models. This highlights the substantial improvement in precision gained by incorporating heterogeneity into the analysis.
\begin{figure}[!h]
	\centering
	\includegraphics[width=1\linewidth]{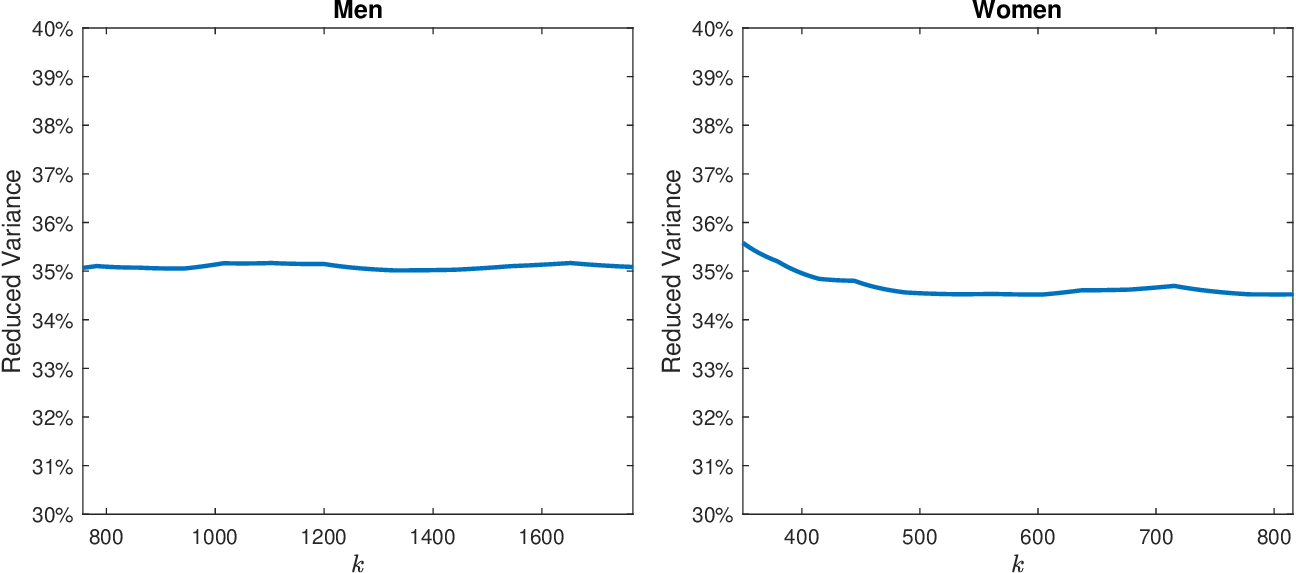}
	\caption{Percentage of reduced variance when accounting for athlete heterogeneity.}
	\label{fig:varlosskep}
\end{figure}

\begin{figure}[!h]
	\centering
	\includegraphics[width=1\linewidth]{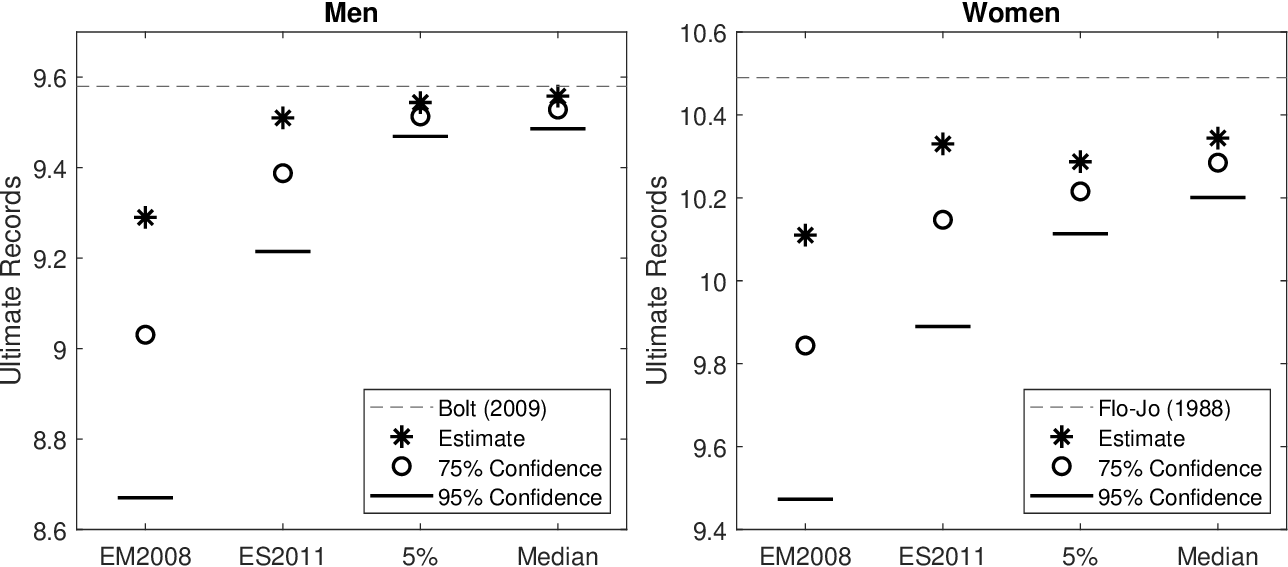}
	\caption{Comparison of our estimates (stars), 75\% (circles), and 95\% (lines) lower confidence bounds with previous ones. Dashed lines represent the current world records.}
	\label{fig:boxmaleandfemale}
\end{figure}
Figure \ref{fig:boxmaleandfemale} compares the results of  our analysis with those calibrated from the procedures in EM2008 and \cite{ES2011}. 
 To ensure a robust analysis, we present in the `5\%' and `Median' columns, the point estimates as described above and similarly two different confidence bounds: (1) the results using $k/n \approx 5\%$, as in Figure \ref{fig:linearplotmaleandfemale}, and (2) the median of the confidence bounds across $k$ values ranging from 3\% to 7\% of the sample size, as in Figure \ref{fig:fitk}. The confidence bounds  are given for the levels  75\% and 95\% for the four procedures. 

For both approaches, our lower confidence bounds are much closer to the point estimates (and substantially higher) than those reported in EM2008 and \cite{ES2011}, due to the larger dataset analyzed and the improvements achieved by accounting for heterogeneity across athletes. Using $k$ around 5\% of the sample size, the 95\% lower confidence limits for the ultimate records are estimated at 9.47 seconds for men and 10.11 seconds for women. The median approach slightly raises these bounds to \textbf{9.49} seconds for men and more substantially to \textbf{10.20} seconds for women.

\section{Summary and Conclusion}
We estimate the ultimate world records for men and women on the 100-meter dash using the novel theory of heterogeneous extreme value statistics. For men we estimate the ultimate world record to be 9.56 seconds. The present world record 9.58 of Usain Bolt in 2009 is very close to this estimate. Our estimate can alternatively be seen as a 50\% lower confidence bound on the ultimate world record, meaning that it is well possible that this estimated ultimate world record can be broken. More prudently, employing this heterogeneous extreme value theory and using multiple times per athlete we construct a very accurate asymptotic 95\% lower confidence bound on the ultimate world record of 9.49 seconds, still quite close to the present world record, that is, under the present conditions not much improvement is possible. For the women's 100-meter dash our point estimate of the ultimate world record is 10.34 seconds, again rather close to the very old 1988 world record 10.49 of Florence Griffith Joyner. The more prudent 95\% lower confidence bound on the women's ultimate world record is 10.20. The larger margin for women  of  0.29 seconds is partly due to the smaller sample size compared with men.

\appendix
\section{Appendix}
	Recall the notation from Subsection \ref{het}. One can extend the definition of $\lambda$ in  \eqref{lambda-func:general} to define the  function:
	\begin{equation}\label{defr}
		R(x,y)=\lim\limits_{n \rightarrow \infty} \frac{1}{k} \sum_{i=1}^{n} P\left( U_{i}^{(n)} < \frac{k}{n}x\right) P\left( U_{i}^{(n)} < \frac{k}{n} y \right),
	\end{equation}
	such that
	\begin{equation*}
		\lambda(u) \equiv R(1, 1/u),\quad u>0.
	\end{equation*}
	This function  $R$ exhibits the following  properties:
	\begin{enumerate}
		\item Monotonicity:  
		The  function $R$ is increasing in each of its coordinates.
		\item Homogeneity:  $R(ax, ay) = aR(x, y)$, for all $a, x, y > 0$.
		\item Symmetry:  $R(x, y) = R(y, x)$, for all $x, y > 0$.
	\end{enumerate}
	Similarly, one can extend the definition of the estimators \eqref{eqn:estimator-lambda-individual} and \eqref{eqn:estimator-lambda} as follows:
	\begin{gather*}
		\widehat{G}_{\ell}(x,y)=\frac{1}{m_{\ell}-1} \sum_{1 \leq i \neq j \leq m_{\ell}} \mathds{1} \left[R_{r_\ell+i} > n-kx, \, R_{r_\ell+j} > n-ky \right],\\
		\widehat{R}(x,y)=\frac{1}{k}\sum_{\ell=1}^{n}\widehat{G}_{\ell}(x,y);
	\end{gather*}
	then $\widehat{\lambda}(u)=\widehat{R}(1,1/u)$. 

    We will consider the uniform consistency of $\widehat R$ in the following lemma, which specializes to the uniform consistency of $\widehat\lambda $, given in (\ref{uclh}).

\begin{lem}\label{lem:R-func} Let $F_n$ be continuous and assume $\max_{\ell=1, \dots, p}m_\ell=O(1)$.
	Suppose the limit $R(x,y)$ in \eqref{defr} exists for all possible intermediate sequences in \eqref{k:p}. Now, let $k = k(n)$ be any particular sequence satisfing \eqref{k:p}. 
    Then for each  $M>0$, as $n \to \infty$,
	\begin{equation}\label{convr}
		\sup_{0\leq x,y\leq M}\left|\widehat{R}(x,y)-R(x,y) \right|\xrightarrow{P}0.
	\end{equation}
\end{lem}
\begin{proof}
For (\ref{convr}) we only need to show the pointwise consistency of $\widehat{R}(x,y)$ for $x,y>0$; the uniform consistency follows by the monotonicity of $\widehat{R}$ and $R$ and the homogeneity of $R$ as in the proof of Theorem 7.2.1 in \cite{dHF2006}.
	
	Recall  $U_i^{(n)}=1-F_n(X_i^{(n)})$ and  let $U_{1,n}\leq U_{2,n}\ldots\leq U_{n,n}$ be the order statistics of the  $\{U_i^{(n)}\}$. Observe that, with probability 1, for each $\ell=1, \dots, p$, we have
	$$
		\widehat{G}_\ell(x,y)=\frac{1}{m_{\ell}-1} \sum_{1\leq i\neq j \leq m_{\ell}} 
		\mathds{1}\left[U_{r_\ell+i}^{(n)}<U_{\lceil kx+1\rceil,n},U_{r_\ell+j}^{(n)}<U_{\lceil ky+1\rceil,n} \right]
		= \widetilde G_\ell \left(\frac{n}{k}U_{\lceil kx+1\rceil,n},\frac{n}{k}U_{\lceil ky+1\rceil,n} \right),
	$$
	where
	\begin{equation*}
		\widetilde G_\ell (x,y)=\frac{1}{m_{\ell}-1} \sum_{1\leq i\neq j \leq m_\ell} 
		\mathds{1}\left[U_{r_\ell+i}^{(n)}<kx/n,U_{r_\ell+j}^{(n)}<ky/n \right].
	\end{equation*}
	Hence $$\widehat R(x,y)= \frac{1}{k} \sum_{\ell=1}^{p}
	\widetilde G_\ell \left(\frac{n}{k}U_{\lceil kx+1\rceil,n},\frac{n}{k}U_{\lceil ky+1\rceil,n} \right). $$
	For every $x,y>0$,
	\begin{align*}
		\mathbb{E} \,\frac{1}{k} \sum_{\ell=1}^{p}
		\widetilde G_\ell \left(x,y\right)
		=&\frac{1}{k} \sum_{\ell=1}^{p} \frac{1}{m_{\ell}-1}\sum_{1\leq i\neq j \leq m_{\ell}} P (U_{r_\ell+1}^{(n)}<kx/n) P (U_{r_\ell+1}^{(n)}<ky/n) 
		\\
		=& \frac{1}{k}\sum_{i=1}^{n}P\left(U_i^{(n)}<kx/n \right) P\left(U_i^{(n)}<ky/n \right)
		\rightarrow R(x,y),\end{align*}
	where the last step holds by definition, and
	\begin{align*}
		\text{var}\left(\frac{1}{k} \sum_{\ell=1}^{p}
		\widetilde G_\ell \left(x,y\right) \right)
		=& \frac{1}{k^2}\sum_{\ell=1}^{p}\frac{1}{(m_{\ell}-1)^2}\cdot\text{var}
		\left(\sum_{1\leq i\neq j \leq m_{\ell}}\mathds{1}\left[U_{r_\ell+i}^{(n)}<kx/n,U_{r_\ell+j}^{(n)}<ky/n \right]\right)\\
		\leq&\frac{1}{k^2}\sum_{\ell=1}^{p}\frac{1}{(m_{\ell}-1)^2}\cdot (m_{\ell}(m_{\ell}-1))^2\,P (U_{r_\ell+1}^{(n)}<kx/n) P (U_{r_\ell+1}^{(n)}<ky/n)
		\\
		\leq &\frac{\max_{\ell}m_{\ell}}{k}\,\mathbb{E}\,\frac{1}{k} \sum_{\ell=1}^p
		\widetilde G_\ell \left(x,y\right)\rightarrow 0.
	\end{align*}
	 	This implies that $\frac{1}{k} \sum_{\ell=1}^{n}
	\widetilde G_\ell \left(x,y\right)\xrightarrow{P}R(x,y)$ pointwise. The convergence is then uniform by the continuity of $R$ and the monotonicity of $\frac{1}{k} \sum_{\ell=1}^{n}
	\widetilde G_\ell $ and $R$. 
    
    It remains to show that
	\begin{equation*}\label{eqn:R-consistency-margins}
		\frac{n}{k}U_{\lceil kz+1\rceil,n}\xrightarrow{P}z,
	\end{equation*}
	which is done in the Proof of Theorem 2.3 in the  Supplementary Material of \cite{EH2023}, for positive extreme value index~$\gamma$.  However, it is elementary to  prove it directly, without using $\gamma$. We will omit the details.
\end{proof}

\end{document}